\documentclass[runningheads]{llncs}
\usepackage{graphicx}
\usepackage{mathrsfs}
\usepackage{amssymb}
\usepackage{url}
\usepackage{bbding}
\begin{document}
\title{A Privacy-Preserving Logistics Information System with Traceability}
%
%
\author{Quanru Chen\inst{1} \and Jinguang Han\inst{2(}\Envelope\inst{)} \and Jiguo Li\inst{3,4}, Liquan Chen\inst{5} \and Song Li \inst{1}}
\authorrunning{Chen {\em et al.}}
%
\institute{College of Information Engineering, Nanjing University of Finance and Economics, Nanjing,  China \\
	\email{cqrqx@qq.com, lisong@nufe.edu.cn}
	\and
	Jiangsu Provincial Key Laboratory of E-Business, Nanjing University of Finance and Economics, Nanjing, China \\
	\email{jghan22@gmail.com}
	\and 
	College of Computer and Cyber Security, Fujian Normal University, Fuzhou, China 
	\and  Fujian Provincial Key Laboratory of Network Security and Cryptology, Fuzhou,  China\\
	\email{ljg1688@163.com}
	\and 
	School of Cyber Science and Engineering, Southeast University, Nanjing, China\\
	\email{Lqchen@seu.edu.cn}
	}
\maketitle              
\begin{abstract}
Logistics Information System (LIS) is an interactive system that provides information for logistics managers to monitor and track logistics business. In recent years, with the rise of online shopping, LIS is becoming increasingly important. However, since the lack of effective protection of personal information, privacy protection issue has become the most problem concerned by users. Some data breach events in LIS released users' personal information, including address, phone number, transaction details, etc. In this paper, to protect users' privacy in LIS, a privacy-preserving LIS with traceability (PPLIST) is proposed by combining multi-signature with pseudonym. In our PPLIST scheme, to protect privacy, each user can generate and use different pseudonyms in different logistics services. The processing of one logistics is recorded and unforgeable. Additionally, if the logistics information is abnormal, a trace party can de-anonymize users, and find their real identities. Therefore, our PPLIST efficiently balances the relationship between privacy and traceability.

\end{abstract}
\section{Introduction}
In recent years, with the rapid development of e-commerce, online shopping has become a popular trend. Online shopping is an interactive activity between a buyer and a seller, where after completing an order by a buyer, the product is delivered via a logistics system \cite{8622160}. Logistics system helps to reduce product cost and save shopping time.

Unfortunately, the current LISs \cite{lis} cannot effectively protect users' privacy information. Users' personal information is clearly visible on the express bill and the LIS database \cite{2014QW}. Some data breaches in LISs released users' personal information, including addresses, phone numbers, transaction details, etc. If a user's personal information is leaked and maliciously collected, she may be at high risk of identity forgery and property fraud, in addition to the risk of being harassed by spam messages. Therefore, it is interesting and important to consider the privacy issues in LISs. 

Furthermore, since a product is delivered by multiple logistics stations, it is important to record the whole logistics process and make the process unforgeable. Additionally, to prevent users from conducting illegal transactions, users can be de-anonymized \cite{5207644,7524578} and traced.

In this paper, we propose a privacy-preserving logistics information system with traceability (PPLIST). Compared with the existing LISs, our scheme has the following advantages:
\begin{enumerate}
	\item [1)] Users can anonymously use the logistics services in our PPLIST scheme. Users generate and use different pseudonyms in different logistics services. 
	Even the internal staff of a logistics company can not directly obtain the information of users' identities, our PPLIST effectively protects users' personal information.
	\item [2)] In the case that the identity of a user needs to be released, a trace party can de-anonymize a user and find his identity. This properties prevent users from conducting illegal logistics via a logistics system. 
	\item [3)] Our PPLIST scheme is efficient. Multi-signature is applied to record the delivery process and reduces the storage space. 
\end{enumerate}
\paragraph{Contributions:} Our main contributions in this paper are summarised as follows: 1) The definition and security model of our PPLIST scheme are formalised; 2) A PPLIST scheme is formally constructed; 3) The security of our PPLIST scheme is formally reduced to well-known complexity assumptions; 4) Our PPLIST scheme is implemented and evaluated.
\subsection{Related Work}
In this subsection, we introduce the work which is related to our PPLST scheme, including LIS, privacy protection in LIS, multi-signature and pseudonym.
\subsubsection{Logistics Information System} LIS is a subsystem and the nerve center of logistics systems. As the control center of the whole logistics activities, LIS has many functions. The main functions of LIS are as follows: collect, store, transmit, process, maintain and output logistics information; provide strategic decision support for logistics managers; improve the efficiency of logistics operations \cite{7034998}.

Bardi et al. \cite{bardi1994logistics} pointed out that the choice of LIS directly affected the logistics cost and customer-service level. Lai et al. \cite{lai2005information} showed that LISs is very important for a company to manage product inventory and predict the trend of customers' online shopping. In addition, Ngai et al. \cite{ngai2008logistics} claimed that LIS is an information system that can promote a good communication between the companies and the customers. An LIS adoption model was proposed in \cite{ngai2008logistics} to examine the relationship among organizational environment, perceived benefits and perceived barriers of LIS adoption. In \cite{closs2000logistics}, Closs and Xu argued that the important source of enterprise competitive advantages was logistics information technology. Their research showed that companies with advanced logistics information technology and LIS performed better than other companies.

LISs have been proposed and applied into various application scenarios \cite{Amazon,Taobao}. Amazon \cite{Amazon} is one of the first companies to provide e-commerce services. Amazon has a logistics system, which realizes the organization and operation of the whole logistics activities. Amazon has also added special technology, One-Click \cite{christin2010dissecting}, in their LIS, which can automatically store the information of customers. Therefore, customers do not input their person information in each shopping. In addition, Amazon's LIS has the following functions \cite{7583659}: order confirmation in time, smooth logistics process, accurate inventory information and optional logistics methods, etc. Amazon has become a business to consumer (B2C) e-commerce \cite{kim2005a} company.

Taobao \cite{Taobao} is a consumer to consumer (C2C) e-commerce \cite{5592943} platform. Taobao entrusts all logistics activities to a third party logistics company, but takes a series of measures to ensure the security of logistics activities. For instance, Taobao implements the network real-name system (NRS) \cite{5496530} in their LIS, and has set up a special customer-service department to solve products logistics problems. Besides, Taobao has the functions of timely confirmation of orders and delivery within the specified time.
 
\subsubsection{Privacy Protection in LIS} Although the LIS of e-commerce platform brings convenience to people's life, it also brings great challenges to privacy protection. LIS stores a large number of users' personal information. Once the information is leaked, it will result in serious threaten to the life and safety of users. Some privacy protection methods in LIS have been proposed, such as \cite{2018LIP,011coordinating,2016AN,2016A,tarjan2014a,2014E,xu2011a}. We compare our scheme with these systems in Table 1. 

Léauté et al. \cite{011coordinating} proposed a scheme to ensure the privacy of users while minimizing the cost of logistics operation. The scheme formalizes the problem as a Distributed Constraint Optimization Problem (DCOP) \cite{grinshpoun2013asymmetric}, and combines various techniques of cryptography. But the disadvantage of this scheme \cite{011coordinating} is that the anonymization of users is not considered. In \cite{xu2011a}, Frank et al. proposed a set of protocols for tracking logistics information, which is a light-weight privacy protection mechanism.

To solve the problem of privacy leakage caused by stolen express order number, Wei et al. \cite{2014E} proposed a k-anonymous model to protect logistics information. However, the method only protects a part of users' personal information, because the names and telephone numbers of receivers are directly printed on the express bills for delivery. 

To improve the security of \cite{2014E}, Qi et al. \cite{2016AN} proposed a new logistics management scheme based on encrypted QR code \cite{tarjan2014a}. After a courier scans the encrypted QR code by using an APP, the logistics information of products in the database is automatically updated through GPRS or Wi-Fi. 
The APP provides an optimal delivery route for couriers. However, the problem of \cite{2016AN} is that users' personal information is still visible to the internal staff of express companies. In addition, Laslo et al. \cite{tarjan2014a} proposed a traceable LIS based on QR code. However, this scheme does not consider privacy protection. 

Furthermore, Gao et al. \cite{2018LIP} proposed a secure LIS, named LIP-PA, which can protect the logistics process information between different logistics stations, but the protection of users' personal information is not considered well. Hence, the privacy of users in LISs \cite{2018LIP,2016AN,tarjan2014a} was not fully considered.

Liu et al. \cite{2016A} designed an LIS based on the Near Field Communication (NFC) \cite{2016Research} technology. In \cite{2016A}, users’ personal information was hidden in tags, and only authorized people can access information. However, because of the limitation of computation power, the scheme cannot perform complex encryption and decryption processes.

In summary, above schemes addressed the privacy issues in LIS, but these schemes did not consider the track of delievery process and the trace of illegal users. However, these are important issues in LISs.
Therefore, to solve these problems, we propose a new privacy-preserving LIS called PPLIST.

\begin{table}[!t]\centering 
	\caption{The Comparison between Our Scheme and Related Schemes}\label{tab1}\centering 
	\begin{tabular}{|c|c|c|c|}
		\hline
		Systems & Anonymity & Traceability & Security Proof\\
		\hline
		Gao et al.\cite{2018LIP} & $\times$ & $\times$ & \checkmark\\
		Léauté et al.\cite{011coordinating} & $\times$ & $\times$ & \checkmark\\
		Qi et al.\cite{2016AN}& $\times$ & \checkmark & $\times$ \\
		Liu et al.\cite{2016A}& \checkmark & $\times$ & \checkmark\\
		Laslo et al.\cite{tarjan2014a}& $\times$ & \checkmark & $\times$\\
		Wei et al.\cite{2014E}& \checkmark & $\times$  & $\times$\\
		Frank et al.\cite{xu2011a} & $\times$ & \checkmark & $\times$\\
		Our PPLIST & \checkmark & \checkmark & \checkmark\\
		\hline	
	\end{tabular}
\end{table}

\subsubsection{Multi-Signature} Multi-signature, also called multi-digital signature, is an important branch of digital signature. Multi-signature is suitable to the case where multiple users sign on a message, and a verifier is convinced that each user participated in the signing \cite{inproceedingsBA}.

Itakura \cite{1983A} first proposed the concept of multi-signature, and proposed a multi-signature scheme with fixed number of signatures. Then, many multi-signature schemes were proposed \cite{2006Multi,article27,article33,article36,inproceedings35,article34}. The multi-signature generation time of schemes \cite{1983A,article34} is linear with the number of signers. Okamoto et al. \cite{inproceedings35} proposed a muti-signature scheme, but it, like scheme \cite{article36}, only allows each signer in a group to sign the message. It's inflexible. Furthermore, Ohta and Okamoto \cite{article36} formlized the security model of multi-signature. However, this scheme did not consider the security of the key generation process, so its security is not strong. Based on \cite{article36}, Micali et al. \cite{article33} proposed a formal and strong security model for multi-signature. Bellare and Neven \cite{2006Multi} proposed a new scheme and proved its secure in the plain public-key model. This scheme improved the efficiency of previous multi-signature schemes. 

Since it enables multiple signers to collabratively sign on a message, multi-signature has been used into various application scenarios, such as \cite{5714258,2013Efficient,8187818,2018Compact}. Shacham \cite{2003Sequential} proposed a sequential aggregate multi-signature scheme. The scheme computed the final multi-signature by sequentially aggregating the signatures from multiple signers. However, the data transmission of \cite{2003Sequential} is large. To solve this problem, Neven \cite{5714258} presented a new sequential aggregate multi-signature scheme based on \cite{2003Sequential}. The scheme of \cite{5714258} reduces signing and verification costs effectively. 

Tiwari et al. \cite{2013Efficient} proposed a secure multi-proxy multi-signature scheme. It does not need paring operations, and reduces the running time. The scheme is aslo secure against the attack of selected messages. However, Asaar et al.\cite{8187818} found the scheme in \cite{2013Efficient} is insecure, and proposed an identity-based multi-proxy and multi-signature scheme without pairing. The security of this scheme was reduced to the RSA assumption in the random oracle model by using the Forking Lemma technique \cite{2008MultisignaturesS}. 

Recently, Dan et al.\cite{2018Compact} proposed a new multi-signature scheme. Signature compression and public-key aggregation were used in the scheme. Therefore, when a group of signers signed a message, the verifier only needs to verify the final aggregate signature. The advantage of this scheme is that the size of final aggregate signature is constant and independent of the number of signers. Furthermore, this scheme is secure against rogue-key attacks. When constructing our PPLIST, we apply the scheme \cite{2018Compact} to record the whole logistics process and reduces the storage cost.

\subsubsection{Pseudonym} Pseudonym is a method that allows users to interact anonymously with other organizations. Because pseudonym is unlinkable, it can effectively protect the information of a user's identity \cite{8965262} among multiple authentications. The common pseudonym generation techniques are as follows \cite{MC2019}: 1) Encryption with public key; 2) Hash function; 3) Keyed-hash function with stored key; 4) Tokenization. 

Chaum \cite{1985DC} found that pseudonym enables users to work anonymously with multiple organizations, and users can use different pseudonyms in different organizations. Because of the unlinkability of pseudonym, no organization can link a user's pseudonyms to her identities. Later, Chaum and Evertse \cite{inproceedingsC} presented a pseudonym model scheme based on RSA. However, the scheme needs a trusted center to complete the sign and transfer of all users' credentials.

To reduce the trust on the trusted center, Chen \cite{inproceedingsCL} proposed a scheme based on the discrete logarithm assumption. The scheme also needs a trusted center, but the trusted center is only required for pseudonym verification. Although Chen's scheme is less dependent on the trusted center than the scheme \cite{inproceedingsC}, the trusted center was still required.

In order to enable users to have the initiative in the pseudonym system, Lysyanskaya et al. \cite{articleLA} proposed a new scheme. In this scheme, a user's  master secret key was introduced. If the master secret keys are different, the information of users' identities must be different. In addition, the pseudonym certificate submitted by a user to an organization only corresponds to the user's master public key and does not disclose the information of his master secret key.

Pseudonym has been applied in some schemes \cite{2020Anonymous,8103347} to protect users' privacy. To reduce the communication cost of traditional pseudonym systems in Internet of Vehicles, Kang et al. \cite{8103347} proposed a privacy-preserved pseudonym scheme. In this scheme, the network edge resources were used for effective management, and the communication cost was effectively reduced.

In \cite{2020Anonymous}, Han et al. proposed an anonymous single sign-on (ASSO) scheme. In this scheme, pseudonym was applied to protect users' identities. A user uses his secret key to generate different pseudonyms, and obtains a ticket from a ticket issuer anonymously without releasing anything about his real identity. Furthermore, a user can use different pseudonyms to buy different tickets and the ticket issuer cannot know whether two tickets are for the same user or two different users. In our PPLIST scheme, to protect users' privacy, we apply the pseudonym developed in \cite{2020Anonymous} to enable users to use logistics services anonymously and unlinkably.
\subsection{Paper Organisation}
The remainder of this paper is organised as follows. Section 2 presents the preliminaries used in our scheme, and describes the formal definition and security model of our PPLIST scheme. Section 3 provides the construction of our scheme. The security proof and implementation of our scheme are presented in Section 4 and Section 5, respectively. Finally, Section 6 concludes this paper.
\section{Preliminaries}
In this section, the preliminaries used throughout this paper are introduced, including bilinear group, complexity assumptions, formal definition and security model. Table 2 summaries the notations used in this paper

\begin{table}\centering
	\caption{Notation Summary}\label{tab2}\centering
	\begin{tabular}{c|l||c|l}
		\hline
		Notation & Explanation &Notation & Explanation\\
		\hline
		$1^{l}$ & A security parameter & \emph{Pseudonym} & The pseudonym of $U$ \\
		$S_{i}$ & The \emph{i}-th logistics station  & $PUB$ & Public parameters\\
		$U$ & User & PPT & Probable polynomial-time  \\
		$T$ & The trace party & $\mathscr{B}(1^{l})$ & A bilinear group generator\\
		\emph{YA} & The aggregation of $AgY$ & $x\stackrel{R}{\leftarrow}X$ & \emph{x} is randomly selected from  $X$  \\
		\emph{AgY} & A set of selected public keys & $H_{1}, H_{2},H_{3}$ & Cryptographic hash functions \\
		$\sigma$ & The aggregation of signatures & $Sig_{i}$ & The \emph{i}-th single signature  \\
		$\pi$ & The proof of user's ownership & \emph{I} & A set consisting of the indexes  \\
		\emph{d} & The number of elements in  $I$ & & of selected logistics stations\\
		\emph{q} & A prime number\\
		\hline
	\end{tabular}
\end{table}

The framework of our PPLIST is presented in Fig.~\ref{fig1}. The system first generates the public parameters $PUB$. Then, each entity (e.g. logistics station, user and the trace party) generates its secret-public key pair. Prior to ordering a service, the user generates a pseudonym by using his secret key. The system determines the delivery path, and then generates the aggregated public key of the selected logistics stations. After that, each selected logistics station $S_{i}$ generates its single signature $Sig_{i}$ on the product information, pseudonym and aggregated public key, and then passes it to the next selected logistics station. Finally, the last selected logistic station generates its signature and the aggregate signature $\sigma$.  To obtain a product, the user needs to prove that he is the owner by generating a proof of the knowledge included in the pseudonym. The user can verify whether the product is delivered correctly by checking the aggregate signature $\sigma$. In the case that the identity of a user needs to be traced, the trace party can use his secret key to de-anonymous the pseudonym, and find the user's identity. 
\begin{figure}
	\includegraphics[width=\textwidth]{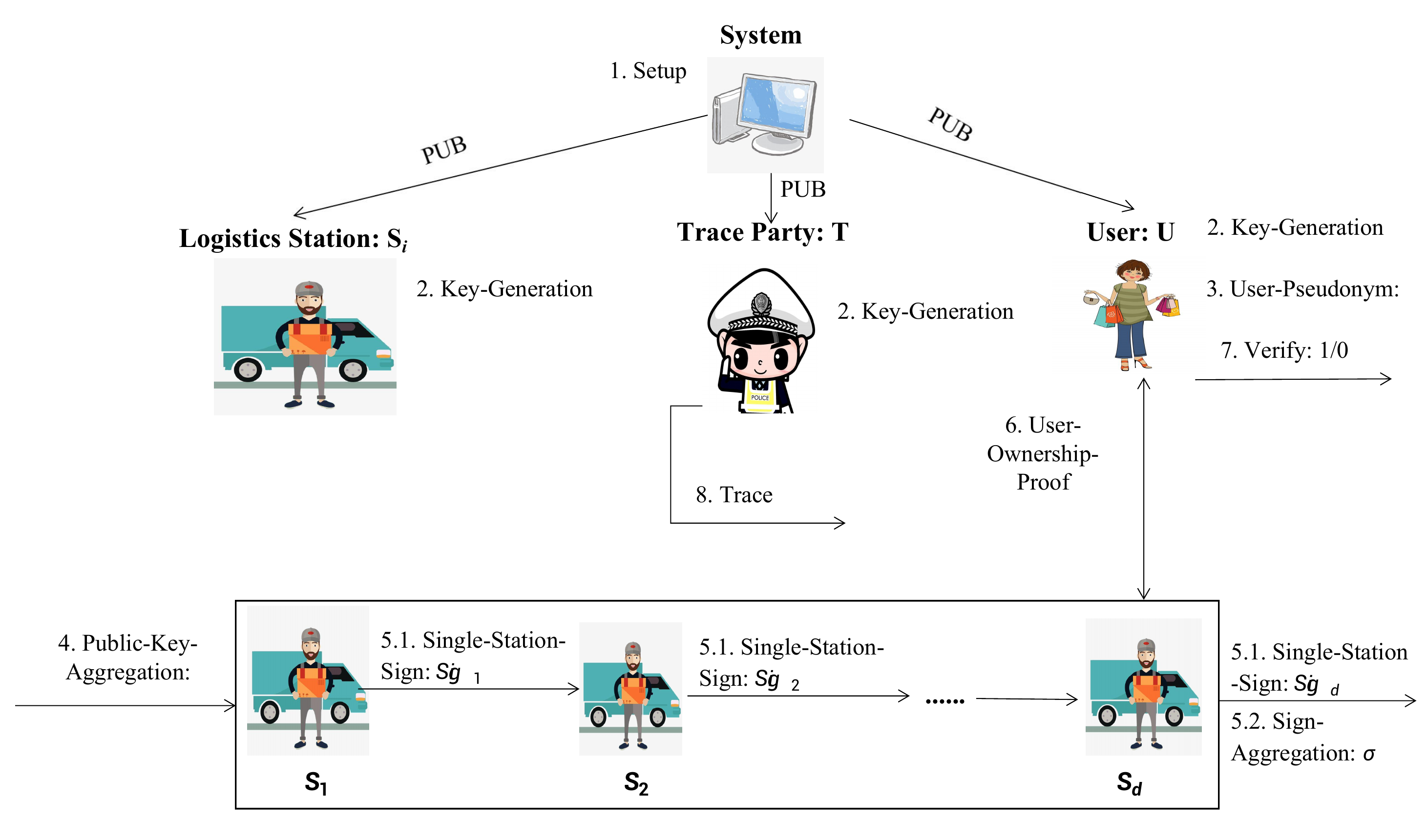}
	\caption{The Framework of Our PPLIST Scheme} \label{fig1}
\end{figure}

\subsection{Bilinear Group}
Let $G_{1},G_{2},G_{\iota}$ be cyclic groups with prime order $q$. A map $e : G_{1}\times G_{2}\rightarrow G_{\iota}$ is a bilinear map if it satisfies the following properties: \textbf{(1) Bilinearity:} For all $g\in G_{1}$, $h\in G_{2}$, $a,b\in Z_{q}$, $e(g^{a},h^{b})=e(g^{b},h^{a})=e(g,h)^{ab}$; \textbf{(2) Non-degeneration:} For all $g\in G_{1}$, $h\in G_{2}$, $e(g,h)\not=1_{\iota}$, where $1_{\iota}$ is the identity element in $G_{\iota}$; \textbf{(3) Computability:} For all $g\in G_{1}$, $h\in G_{2}$, there exists an efficient algorithm to compute $e(g,h)$.

Let $\mathscr{B}(1^{l})\rightarrow(e,q,G_{1},G_{2},G_{\iota})$ be a bilinear group generator which takes as input a security parameter $1^{l}$ and outputs a bilinear group $(e,q,G_{1},G_{2},G_{\iota})$.

A function $\epsilon(x)$ is negligible if for any $k\in N$, there exist a $z\in N$ such that $\epsilon(x)<\frac{1}{x^{z}}$ when $x>z$.
\subsection{Complexity Assumptions}
\begin{definition}[Computational Diffie-Hellman (CDH) Assumption \cite{2008Efficient}] Let $\mathscr{B}(1^{l}) \rightarrow(e,q,G_{1},G_{2},G_{\iota})$, and $g_{1},g_{2}$ be generator of $G_{1},G_{2}$, respectively. Suppose that $\alpha,\beta \stackrel{R}{\leftarrow} Z_{q}$. Given a triple $(g^{\alpha}_{1},g^{\beta}_{1},g^{\beta}_{2})$, we say that the $CDH$ assumption holds on $(e,q,G_{1},G_{2},G_{\iota})$ if all $PPT$ adversaries $\mathcal{A}$ can output $g^{\alpha \beta}_{1}$ with a negligible advantage, namely $Adv_{\mathcal{A}}^{CDH}=PR[\mathcal{A}(g^{\alpha}_{1},g^{\beta}_{1},g^{\beta}_{2})\rightarrow g^{\alpha \beta}_{1}]\leq \epsilon(l)$.
\end{definition}
\begin{definition}[Discrete Logarithm (DL) Assumption \cite{Gordon1993Discrete}] Let $G$ be a cyclic group with prime order $q$, and $g$ be a generator of $G$. Given $Y \in G$, we say that the $DL$ assumption holds on $G$ if all $PPT$ adversaries can output a number $x\in Z_{q}$ such that $Y=g^{x}$ with a negligible advantage, namely $Adv_{\mathcal{A}}^{DL}=PR[Y=g^{x}|\mathcal{A}(q,g,G,Y)\rightarrow x]\leq \epsilon(l)$.
\end{definition}
\subsection{Formal Definition}
A PPLIST scheme is formalized by the following eight algorithms:

\paragraph{$Setup(1^{l})\rightarrow PUB.$} The algorithm takes the security parameters $1^{l}$ as input and outputs the public parameters $PUB$.
\paragraph{$Key-Generation.$} This algorithm consists of the following sub-algorithms:
	\begin{itemize}
		\item [1)] 
		$Station-Key-Generation(1^{l})\rightarrow (SK_{S_{i}},PK_{S_{i}}).$ This algorithm is executed by each logistics station $S_{i}$. $S_{i}$ takes the security parameters $1^{l}$ as input, and outputs his secret-public key pair $(SK_{S_{i}},PK_{S_{i}})$, where $i=1,2,3,\cdots,n$.
		\item [2)] 
		$User-Key-Generation(1^{l})\rightarrow (SK_{U},PK_{U}).$ This algorithm is executed by a user $U$. $U$ takes the security parameters $1^{l}$ as input, and outputs his secret-public key pair $(SK_{U},PK_{U})$.
		\item [3)]
		$Trace-Key-Generation(1^{l})\rightarrow (SK_{T},PK_{T}).$ This algorithm is executed by a trace party $T$. $T$ takes the security parameters $1^{l}$ as input, and outputs his secret-public key pair $(SK_{T},PK_{T})$.
	\end{itemize}
\paragraph{$User-Pseudonym(PUB,SK_{U},PK_{T})\rightarrow Pseudonym.$} This algorithm is executed by $U$. $U$ takes as input his secret key $SK_{U}$, the public key $PK_{T}$ of the trace party and the public parameters $PUB$, and outputs a pseudoym $Pseudonym$.
\paragraph{$Public-Key-Aggregation(PUB,PK_{S_{k_{1}}},PK_{S_{k_{2}}}\cdots,PK_{S_{k_{d}}})\rightarrow YA.$} Let $I$ be a set which consists of the indexes of some selected logistics stations. This algorithm takes as input the public parameters $PUB$ and the public keys $PK_{S_{k_{1}}},\\PK_{S_{k_{2}}},$ $\cdots,PK_{S_{k_{d}}}$ of selected logistics stations, and outputs the aggregated public key $YA$.
\paragraph{$Sign.$} This algorithm consists of the following sub-algorithms:
	\begin{itemize}
		\item [1)] $Single-Station-Sign(PUB,SK_{S_{k_{i}}},Pseudonym,YA,m)\rightarrow Sig_{i}.$ This algorithm is executed by each selected logistics station $S_{k_{i}}$. $S_{k_{i}}$ takes as input its secret key $SK_{S_{k_{i}}}$, the aggregated public key $YA$, product information $m$ and the public parameters $PUB$, and outputs a signature $Sig_{i}$, where $i=1,2,3,\cdots,d$.
		\item [2)]
		$Sign-Aggregation(PUB,Sig_{1},Sig_{2},\cdots,Sig_{d})\rightarrow \sigma.$ This algorithm takes as input the public parameters $PUB$ and signatures $Sig_{i}$, and outputs a final signature $\sigma$.
	\end{itemize}
\paragraph{$User-Ownership-Verify(PUB,SK_{U},PK_{T},Pseudonym)\leftrightarrow S_{i}(PUB)\rightarrow (\pi,1/0).$} This algorithm is executed between $S_{i}$ and $U$. 
	\begin{itemize}
		\item [1)]
		$U$ takes as input his secret key $SK_{U}$, the $T's$ public key $PK_{T}$, his pseudonym $Pseudonym$ and the public parameters $PUB$, and outputs a proof $\pi$. 
		\item [2)]
		The verifier takes as input the public parameters $PUB$, and outputs $1$ if the proof $\pi$ is valid; otherwise, it outputs $0$ to show the proof is invalid.
	\end{itemize}
\paragraph{$Verify(PUB,\sigma,Pseudonym,YA,m)\rightarrow 1/0.$} This algorithm takes as input the public parameters $PUB$, the final signature $\sigma$, the pseudonym $Pseudonym$, the aggregated public key $YA$ and product information $m$, and outputs $1$ if signature is valid; otherwise, it outputs $0$ to show it is an invalid signature.
\paragraph{$Trace(PUB,\sigma,SK_{T},Pseudonym,YA,m)\rightarrow PK_{U}/\bot.$} This algorithm is executed by $T$. $T$ takes as input his secret key $SK_{T}$, the pseudonym $Pseudonym$, the aggregated public key $YA$, the final signature $\sigma$, product information $m$ and the public parameters $PUB$, and outputs $U's$ public key $PK_{U}$ if the signature $\sigma$ is valid; otherwise, it outputs $\bot$ to show failure.
\subsection{Security Requirements}
The security model of our scheme is defined by the following two games.
\subsubsection{Unforgeability.} This is used to define the unforgeability of signature, namely even if users, the trace party and the other stations collude, they cannot forge a valid signature on behalf of the selected logistics stations. This game is executed between a challenger $\mathcal{C}$ and a forger $\mathcal{F}$.
\paragraph{Setup.} $\mathcal{C}$ runs $Setup(1^{l})\rightarrow PUB$ and sends $PUB$ to $\mathcal{F}$.
\paragraph{Key-Generation Query.} 
\begin{enumerate}
	\item [1)] $\mathcal{F}$ asks the public key of stations. $\mathcal{C}$ runs $Station-Key-Generation(1^{l})\rightarrow (SK_{S_{i}},PK_{S_{i}})$ and sends the station's public key $PK_{S_{i}}$ to $\mathcal{F}$.
	\item [2)] When $\mathcal{F}$ asks a urse's secret-public key pair, $\mathcal{C}$ runs $User-Key-Generation\\(1^{l})\rightarrow (SK_{U},PK_{U})$ and sends $(SK_{U},PK_{U})$ to $\mathcal{F}$. Let $GPU$ be a set of users' public key.
	\item [3)] When $\mathcal{F}$ asks the secret-public key of the trace party, $\mathcal{C}$ runs $Trace-Key-Generation(1^{l})\rightarrow (SK_{T},PK_{T})$ and sends $(SK_{T},PK_{T})$ to $\mathcal{F}$.
\end{enumerate}
\paragraph{User-Pseudonym Query.} $\mathcal{F}$ submits a $SK_{U}$ and the public key $PK_{T}$ of the trace party, $\mathcal{C}$ runs $User-Pseudonym(PUB,SK_{U},PK_{T})\rightarrow Pseudonym$ and sends $Pseudonym$ to $\mathcal{F}$. Let $UPQ$ be a set of pseudonyms of users.
\paragraph{Public-Key-Aggregation Query.} Let $I$ be a set which consists of the indexes of some selected logistics stations and let $d$ be the number of elements in the set $I$. $\mathcal{F}$ submits a group of selected stations' public keys. $\mathcal{C}$ runs
$Public-Key-Aggregation(PUB,PK_{S_{k_{i}}})\rightarrow YA$, where $i=1, 2,\cdots,d$. $\mathcal{C}$ returns $YA$ to $\mathcal{F}$.
\paragraph{Sign Query.} $\mathcal{F}$ adaptively submits selected station's secret key $SK_{S_{k_{i}}}$, the aggregation of public key $YA$, and $U's$ pseudonym $Pseudonym$ and the product information $m$ to ask for a single signature $Sig_{i}$ up to $\varrho$ times.
\paragraph{Output.}  $\mathcal{F}$ outputs a forged signature $Sig_{i}^{'}$, a final signature $\sigma^{'}$, $U's$ pseudonym $Pseudonym$ and the product information $m'$, the public keys of selected logistics stations $AgY$ and the aggregated public keys $YA^{'}$. $\mathcal{F}$ wins the game if $PK_{S_{i}}\in AgY$, $\mathcal{F}$ has not conducted signature query on the message $m'$, and $Verify(PUB,\sigma^{'},Pseudonym,YA^{'},m')=1$.
\begin{definition}
	A privacy-preserving logistics information system with traceability  is $(\varrho,\epsilon(l))$ unforgeable if all probabilistic
	polynomial-time (PPT) forger $\mathcal{F}$ who makes $\varrho$ signature queries can only win the above game with a negligible
	advantage, namely
	\begin{equation}
		Adv=Pr\Bigg[Verify(PUB,\sigma^{'},Pseudonym,YA^{'},m')=1\Bigg]\leq\epsilon(l)
	\end{equation}
\end{definition}
\subsubsection{Traceability.} This is used to formalise the traceability of our scheme, namely an attacker $\mathcal{A}$ cannot  frame a user who did not use the logistics services. We suppose that at least one station is honest. This game is executed between a challenger $\mathcal{C}$ and an attacker $\mathcal{A}$.
\paragraph{Setup.} $\mathcal{C}$ runs $Setup(1^{l})\rightarrow PUB$ and sends $PUB$ to $\mathcal{A}$.
\paragraph{Key-Generation Query.} 
\begin{enumerate}
	\item [1)] $\mathcal{A}$ can ask for the public key of each station. $\mathcal{C}$ runs $Station-Key-Generation(1^{l})\rightarrow (SK_{S_{i}},PK_{S_{i}})$ and sends the station's public key $PK_{S_{i}}$ to $\mathcal{F}$.
	\item [2)] When $\mathcal{A}$ asks a urse's secret-public key pair, $\mathcal{C}$ runs $User-Key-Generation\\(1^{l})\rightarrow (SK_{U},PK_{U})$. Let the secret-public key pair of $U^{*}$ be $(SK_{U^{*}},PK_{U^{*}})$. $\mathcal{C}$ sends other users' secret-public key pair $(SK_{U},PK_{U})$ and $PK_{U^{*}}$ to $\mathcal{A}$. Let $GPU$ be a set consisting of  users's public keys.
	\item [3)] When $\mathcal{A}$ asks the secret-public key pair of the trace party, $\mathcal{C}$ runs $Trace-Key-Generation(1^{l})\rightarrow (SK_{T},PK_{T})$ and sends $(SK_{T},PK_{T})$ to $\mathcal{A}$.
\end{enumerate}
\paragraph{User-Pseudonym Query.} $\mathcal{A}$ submits a user's $SK_{U}$ and the public key $PK_{T}$ of the trace party, $\mathcal{C}$ runs $User-Pseudonym(PUB,SK_{U},PK_{T})\rightarrow Pseudonym$ and sends $Pseudonym$ to $\mathcal{A}$. Let $UPQ$ be a set of pseudonyms of users.
\paragraph{Public-Key-Aggregation Query.} Let $I$ be a set which consists of the indexes of some selected logistics stations and let $d$ be the number of elements in the set $I$. $\mathcal{A}$ submits a group of selected stations' public keys. $\mathcal{C}$ runs
$Public-Key-Aggregation(PUB,PK_{S_{k_{i}}})\rightarrow YA$, where $i=i,2,\cdots,d$. $\mathcal{C}$ returns $YA$ to $\mathcal{A}$.
\paragraph{Sign Query.} $\mathcal{A}$ adaptively submits a selected station's secret key $SK_{S_{k_{i}}}$, the aggregation of public key $YA$, and $U's$ pseudonym $Pseudonym$ and the product information $m$ to ask for a single signature $Sig_{i}$ up to $\varrho$ times.
\paragraph{Output.}  $\mathcal{A}$ outputs a tuple $(\sigma',Pseudonym^{'},YA',m')$. $\mathcal{A}$ wins the game if $Trace(PUB,\sigma^{'},SK_{T},Pseudonym^{'},YA',m')\rightarrow PK_{U^{*}}^{'}$ with $PK_{U^{*}}^{'}\notin GPU$ or $PK_{U^{*}}^{'}\not=PK_{U^{*}}\in GPU$.
\begin{definition}
	A privacy-preserving logistics information system with traceability is $(\varrho,\epsilon(l))$ traceable if all probabilistic polynomial-time (PPT) adversary $\mathcal{A}$ who makes $\varrho$ signature queries can only win the above game with a negligible
	advantage, namely
	\begin{equation}
		\begin{array}{c|c}
			Adv=Pr\Bigg[
			\begin{array}{c}
				PK_{U^{*}}^{'}\notin GPU \ or\\ PK_{U^{*}}^{'}\not=PK_{U^{*}}\in GPU
			\end{array} &
			\begin{array}{c}
				Trusted-Party-Trace\\(PUB,\sigma^{'},SK_{T},Pseudonym^{'},\\YA',m')\rightarrow PK_{U^{*}}^{'}
			\end{array}
			\Bigg]\leq\epsilon(l)
		\end{array}
	\end{equation}
\end{definition}
\section{Construction of Our Scheme}
In this section, we introduce the construction of our scheme. We firstly present a high-level overview, and then describe the formal construction of our scheme.
\subsection{High-Level Overview}
The high-level overview of our scheme is as follows.
\paragraph{Setup.} The system generates the corresponding public
parameters $PUB$.
\paragraph{Key-Generation.} Suppose that there are $n$ logistics stations. Each $S_{i}$, $U$ and $T$ generate their secret-public key pairs $(x_{s_{i}},Y_{s_{i}})$, $(x_{u},Y_{u})$ and $(x_{t},Y_{t})$, where $i=1,2,3,\cdots,n$. 
\paragraph{User-Pseudonym.} In order to protect privacy in a delivery process, $U$ generates a pseudonym $Pseudonym$  by using his secret key $SK_{U}$ and $T's$ public key $PK_{T}$.    
\paragraph{Public-Key-Aggregation.} According to product information, the system determines the logistics process by selecting a set of logistics stations $S_{k_{1}},S_{k_{2}},\cdots,S_{k_{d}}$. Let $AgY=\big\{Y_{S_{k_{1}}},Y_{S_{k_{2}}},\cdots,Y_{S_{k_{d}}} \big\}$ be a set consisting of the public keys of the selected logistics stations. For each service, a table $Table$ is built to record its delivery information. The system uses the public key of each $S_{k_{i}}$ and the set $AgY$ to generate $h_{i}$ to resist the rogue key attacks, where $i=1,2,3,\cdots,d$. Then, the system generates the aggregated public key $YA$.
\paragraph{Sign.} Each selected logistics station $S_{k_{i}}$ uses his secret key $x_{s_{k_{i}}}$ to generate a signature $Sig_{i}$ on $U's$ pseudonym $Pseudonym$ and the product information $m$, and sends $Sig_{i}$ to the next logistics stations. Finally, the last logistic station $S_{k_{d}}$ use his secret key $x_{k_{d}}$ to generate a single signature $Sig_{d}$ on  $U's$ pseudonym $Pseudonym$ and the producte information $m$, and compute the aggregated signature $\sigma=\prod_{i=1}^{d}Sig_{i}$. $S_{k_{d}}$ also adds $\sigma$ to the table $Table$.
\paragraph{User-Ownership-Verify.} When $U$ proves to the last logistic station $S_{k_{d}}$ that he is the owner of the product, he proves that his secret key $x_{u}$ is included in the pseudonym $Pseudonym$ by executing a zero-knowledge proof with $S_{k_{d}}$. If the proof is correct, $U$ is the owner of the product; otherwise, he is not the owner of the product.
\paragraph{Verify.} When $U$ receives a product, he checks whether the product was delivered correctly by checking the validity of the aggregate signature $\sigma$. If it is, the product is delivered correctly; otherwise, there are some problems in the delivery process.
\paragraph{Trace.} Given $(\sigma,Pseudonym,AgY,m)$, in the case that a user needs to be de-anonymized, the trace party $T$ first checks whether the signature is correct or not. If it is incorrect, $T$ aborts; otherwise, $T$ users his secret key $x_{t}$ to de-anonymize the Pseudonym and get $U$'s public key $Y_{u}$. 
\subsection{Formal Construction}
The formal construction of our PPLIST scheme is formalised by the following eight algorithms:
\paragraph{Setup.} The system runs $\mathscr{B}(1^{l})\rightarrow(e,q,G_{1},G_{2},G_{\iota},g_{1},g_{2})$ with $e : G_{1}\times G_{2}\rightarrow G_{\iota}$. Let $g_{1}$ be a generator of $G_{1}$ and $g_{2}$ be a generator of $G_{2}$. Suppose that $H_{1}:\big\{0,1\big\}^{*}\rightarrow G_{1},
H_{2}:\big\{0,1\big\}^{*}\rightarrow Z_{q}$ and $H_{3}:\big\{0,1\big\}^{*}\rightarrow Z_{q}$ are cryptographic hash functions. The public parameters are $PUB=(e,q,G_{1},G_{2},G_{\iota},g_{1},g_{2},H_{1},H_{2},H_{3})$.
\paragraph{Key-Generation.}
\begin{itemize}
	\item [1)] $Station-Key-Generation.$ Each logistics station $S_{i}$ selects $x_{s_{i}}\stackrel{R}{\leftarrow}Z_{q}$ and computes $Y_{s_{i}}=g_{2}^{x_{s_{i}}}$. The secret-public key pair of $S_{i}$ is $(x_{s_{i}},Y_{s_{i}})$, where $i=1,2,3,\cdots,n$.
	\item [2)] $User-Key-Generation.$ Each $U$ selects $x_{u}\stackrel{R}{\leftarrow}Z_{q}$ and computes $Y_{u}=g_{2}^{x_{u}}$. The secret-public key pair of $U$ is $(x_{u},Y_{u})$.
	\item [3)] $Trace-Key-Generation.$ $T$ selects $x_{t}\stackrel{R}{\leftarrow}Z_{q}$ and computes $Y_{t}=g_{2}^{x_{t}}$. The secret-public key pair of $T$ is $(x_{t},Y_{t})$.
\end{itemize}
\paragraph{User-Pseudonym.} To generate a pseudonym for a product information $m$, $U$ firstly computes $k=H_{3}(x_{u}\parallel m)$ and then computes $C_{1}=g_{2}^{k}, C_{2}=Y_{t}^{k}\cdot g_{2}^{x_{u}}$. The pseudonym is $Pseudonym=(C_{1},C_{2})$.
\paragraph{Public-Key-Aggregation.} Let $AgY=\big\{Y_{S_{k_{1}}},Y_{S_{k_{2}}},\cdots,Y_{S_{k_{d}}} \big\}$ be a set consisting of the public keys of the logistics stations which will deliver the product to the user. The system firstly computes $h_{i}=H_{2}(Y_{S_{k_{i}}}\parallel AgY)$, and then computes $YA=\prod_{i=1}^{d}{Y_{S_{k_{i}}}}^{h_{i}}$. Let $(Pseudonym,m,AgY,YA)$ be a record of the product information $m$. The system adds it into the table $Table$.
\paragraph{Sign.} When receiving a product, each $S_{k_{i}}$ computes $Sig_{i}=H_{1}(C_{1}\parallel C_{2}\parallel m)^{h_{i}\cdot x_{S_{k_{i}}}}$. $S_{k_{i}}$ sends $Sig_{i}$ to $S_{k_{i+1}}$ for $i=1,2,3,\cdots,d-2$. Finally, $S_{k_{d}}$ computes $Sig_{d}$ and $\sigma = \prod_{i=1}^{d}Sig_{i}$. Subsequently, $S_{k_{d}}$ adds it into the record of $m$ in the table $Table$.
\paragraph{User-Ownership-Verify.} To prove the ownership of the product to the last logistics station $S_{k_{d}}$. $U$ and $S_{k_{d}}$ work as follows.
\begin{itemize}
	\item [1)] $U$ selects $v_{1}\stackrel{R}{\leftarrow}Z_{q},v_{2}\stackrel{R}{\leftarrow}Z_{q}$ and computes $V_{1}=g_{2}^{v_{1}},V_{2}=Y_{t}^{v_{1}}\cdot g_{2}^{v{2}}$. 
	\item [2)] $U$ sends $C_{1},C_{2},V_{1},V_{2}$ to $S_{k_{d}}$. $S_{k_{d}}$ selects $c\stackrel{R}{\leftarrow}Z_{q}$, and returns it to $U$. 
	\item [3)] $U$ computes $r_{1}=v_{1}-c\cdot k$, and $r_{2}=v_{2}-c\cdot x_{u}$, and returns $(r_{1},r_{2})$ to $S_{k_{d}}$. 
	\item [4)] $S_{k_{d}}$ verifies $V_{1} \stackrel{?}{=}g_{2}^{r_{1}}\cdot C_{1}^{c}$, and $V_{2} \stackrel{?}{=}Y_{t}^{r_{1}}\cdot g_{2}^{r_{2}}\cdot C_{2}^{c}$. If these equations hold, it outputs $1$ to show that $U$ is the owner of the product; otherwise, it outputs $0$ to show that $U$ is not the owner of the product. 
\end{itemize}
\paragraph{Verify.} $U$ verifies $e(\sigma,g_{2}^{-1})\cdot e(H_{1}(C_{1}\parallel C_{2}\parallel m),YA) \stackrel{?}{=}1_{G_{\iota}}$. If the equation holds, it outputs $1$ to show that the delivery process is correct; otherwise, it outputs $0$ to show that there are some errors in the delivery.
\paragraph{Trace.} In the case that the identity of $U$ who selected the product $m$ needs to be revealed, $T$ searches in the table $Table$, and finds the record $(\sigma,Pseudonym,YA,$ $m)$ firstly. Then, $T$ verifies $e(\sigma,g_{2}^{-1})\cdot e(H_{1}(C_{1}\parallel C_{2}\parallel m),YA) \stackrel{?}{=}1_{G_{\iota}}$. If it is not, $T$ quits the system immediately; otherwise, $T$ computes $Y_{u}=C_{2}/C_{1}^{x_{t}}$, and confirms the identity of user.
\section{Security Analysis}
In this section, the security of our scheme is formally proven.
\begin{theorem}\label{theorem1}
	Our privacy-preserving logistics information system with traceability (PPLIST) is $(\varrho,\epsilon(l))-$ unforgeable if and only if the $(\epsilon(l)^{'},T)-$ computational Diffie-Hellman (CDH) assumption holds on the bilinear group $(e,q,G_{1},G_{2},G_{\iota})$ and $H_{1},H_{2}$ are two random oracles and $H_{3}$ is a cryptographic hash function, where $\varrho$ is the number of signature queries made by the forger $\mathcal{F}$, and $\epsilon(l)^{'}\geq \frac{1}{q}\cdot\frac{1}{q-1}\cdot \frac{1}{\varrho}\cdot \epsilon(l)$.	
\end{theorem}
\begin{proof}
	Suppose that there exists a forger $\mathcal{F}$ that can break the unforgeability of our scheme, we can construct an algorithm $\mathcal{B}$ which can use $\mathcal{F}$ to break the CDH assumption. Given $(A,B_{1},B_{2})=(g_{1}^{\alpha},g_{1}^{\beta},g_{2}^{\beta})$, the aim of $\mathcal{B}$ is to output $g_{1}^{\alpha\beta}$.
\end{proof}
\paragraph{Setup.} $\mathcal{B}$ selects $H_{3}:\{0,1\}^{*}\rightarrow \mathbb{Z}_{p}$.
 The public parameters are $PUB=(e,q,G_{1},\\G_{2},G_{\iota},g_{1},g_{2},H_{3})$. 
\paragraph{$\bullet$} $\mathcal{B}$ responds to the queries of $\mathcal{F}$ about the random oracle $H_{1}$. 
 \begin{itemize}
 	\item[1)] $\mathcal{F}$ queries the hash function $H_{1}$ of a pseudonym $(C_{1}^{k},C_{2}^{k})$ and a message $m_{k}$. $\mathcal{B}$ selects $t_{k}\stackrel{R}{\leftarrow}Z_{q}$,and sets $H_{1}(C_{1}^{k}\parallel C_{2}^{k}\parallel m^{K})=g_{1}^{t_{k}}$,where $k=1,2,\cdots n$. $\mathcal{B}$ sends $g_{1}^{t_{i}}$ to $\mathcal{F}$ and adds $(C_{1}^{k},C_{2}^{k},m_{k},g^{t_{k}})$ into the table $Table_{1}$.
 	\item [2)] $\mathcal{F}$ queries the hash function $H_{1}$ of a pseudonym $(C_{1}^{*},C_{2}^{*})$ and a message $m^{*}$. $\mathcal{B}$ sends $g_{1}^{\alpha}$ to $\mathcal{F}$, and adds $(C_{1}^{*},C_{2}^{*},m^{*},g_{1}^{\alpha})$ into the table $Table_{1}$.
 \end{itemize}
\paragraph{$\bullet$}  $\mathcal{B}$ responds to the queries of $\mathcal{F}$ about the random oracle $H_{2}$. Let $j$ is the index of $Y_{s_{j}}\in AgY$.
\begin{itemize}
	\item [1)] when $i\not=j$, $\mathcal{B}$ selects $h_{i}\stackrel{R}{\leftarrow}\mathbb{Z}_{q}$ and sets $h_{i}=H_{2}(Y_{s_{i}}\parallel AgY)$. $\mathcal{B}$ returns $h_{i}$ to $\mathcal{F}$ and adds $(Y_{s_{i}},AgY,h_{i})$ into the table $Table_{2}$. 
	\item [2)] when $i=j$, $\mathcal{B}$ selects $h_{j}\stackrel{R}{\leftarrow}Z_{q}$ and sets $h_{j}=H_{2}(Y_{s_{j}}\parallel AgY)$. $\mathcal{B}$
	 returns $h_{j}$ to $\mathcal{F}$, and adds $(Y_{s_{j}}, AgY, h_{j})$ into the table $Table_{2}$. 
\end{itemize}
\paragraph{Key-Generation\ Query.}
\begin{itemize}
	\item [1)] $Station-Key-Generation\ Query.$
	$\mathcal{B}$ picks a station $S_{j}$ from $S_{1},S_{2},\cdots,S_{n}$. For the $i$-th logistics station key generation query, $\mathcal{B}$ selects $x_{s_{i}}$, and computes $Y_{S_{i}}=g_{2}^{x_{s_{i}}}$ where $i\not=j$. $\mathcal{B}$ returns $Y_{S_{i}}$ to $\mathcal{F}$. For the $j$-th logistics station key generation query, $\mathcal{B}$ returns $B_{2}$ to $\mathcal{F}$.
	\item [2)] $User-Key-Generation\ Query.$
	$\mathcal{B}$ selects $x_{u}\stackrel{R}{\leftarrow}Z_{q}$, and compute $Y_{u}=g_{2}^{x_{u}}$. $\mathcal{B}$ sends the secret-public key pair $(x_{u},Y_{u})$ to $\mathcal{F}$.
	\item [3)] $Trace-Key-Generation\ Query.$
	$\mathcal{B}$ selects $x_{t}\stackrel{R}{\leftarrow}Z_{q}$, and compute $Y_{t}=g_{2}^{x_{t}}$. $\mathcal{B}$ sends the secret-public key pair $(x_{t},Y_{t})$ to $\mathcal{F}$.
\end{itemize}
\paragraph{User-Pseudonym\ Query.}
$\mathcal{F}$ submits a product information $m$. $\mathcal{B}$ computes $k=H_{3}(x_{u}\parallel m)$ firstly, then computes $C_{1}=g_{2}^{k}$, $C_{2}=Y_{t}^{k}\cdot g_{2}^{x_{u}}$. $\mathcal{C}$ sends $(C_{1},C_{2})$ to $\mathcal{F}$.
\paragraph{Public-Key-Aggreation\ Query.}
$\mathcal{F}$ submits a group of logistics stations' public-key and sets $AgY=\big\{ Y_{s_{1}}, Y_{s_{2}}, \cdots, Y_{s_{d}} \big\}$. $\mathcal{B}$ searches in $Table_{2}$, and gets  $h_{i}=H_{2}(Y_{s_{i}}\parallel AgY)$, where $i=1,2,3\cdots d$. $\mathcal{B}$ computes $YA=\prod_{i=1}^{d}Y_{s_{i}}^{h_{i}}$, and sends $YA$ to $\mathcal{F}$. 
\paragraph{Sign\ Query.} $\mathcal{B}$ responds to the queries of $\mathcal{F}$ about the single signature $Sig_{i}$:
\begin{itemize}
	\item [1)] Since $i\not=j$, $\mathcal{F}$ asks about the single signature of the logistics station $S_{i}$ on a pseudonym $(C_{1},C_{2})\not=(C_{1}^{*},C_{2}^{*})$ and the message $m^{*}$. $\mathcal{B}$ computes $Sig_{i}=g_{1}^{\alpha\cdot h_{i}\cdot x_{s_{i}}}$. $\mathcal{B}$ sends $Sig_{i}$ to $\mathcal{F}$.
	\item [2)] Since $i\not=j$, $\mathcal{F}$ asks about the single signature of station $S_{j}$ on a pseudonym $(C_{1},C_{2})$ and a message $m\not=m^{*}$. $\mathcal{B}$ computes $Sig_{i}=g_{1}^{t_{{i}}\cdot h_{i}\cdot \beta}$. $\mathcal{B}$ sends $Sig_{i}$ to $\mathcal{F}$. $\mathcal{F}$ can ask for many times. 
	\item [3)] $\mathcal{F}$ asks about the single signature of station $S_{j}$ on a pseudonym $(C_{1}^{*},C_{2}^{*})$ and the message $m^{*}$. $\mathcal{B}$ aborts.
\end{itemize}
\paragraph{Output.}
$\mathcal{F}$ outputs a forged final signature $\sigma^{'}$. According to the above situations, $\mathcal{F}$ can make $q$ queries of random oracles and $\varrho$ signature generations, respectively. By using Forking lemma technique, for two queries of the random oracle $H_{2}$ on the $j$-th station, $\mathcal{B}$ selects $h_{j}$ and $h_{j}^{'}$ with $h_{j}\not=h_{j}^{'}$. For other selected stations, $\mathcal{B}$ sets $h_{i}$ and $h_{i}^{'}$ with $h_{i}=h_{i}^{'}$. Hence, $YA/YA^{'}=\prod_{i=1}^{d}{Y_{S_{k_{j}}}}^{h_{j}-h_{j}^{'}}$. If $\mathcal{F}$ can forge a valid signature, $\mathcal{B}$ have $Sig_{j}=g_{1}^{\alpha\cdot \beta\cdot h_{i}}$ and $Sig_{j}^{'}=g_{1}^{\alpha\cdot\beta\cdot h_{i}^{'}}$, respectively. Then, $\mathcal{B}$ computes $\sigma/\sigma^{'}=g_{1}^{\alpha\cdot \beta\cdot (h_{i}-h_{i}^{'})}$. Therefore, $\mathcal{B}$ can compute $g_{1}^{\alpha\cdot\beta}=(\sigma/\sigma^{'})^{1/(h_{i}-h_{i}^{'})}$ and break the CDH assumption.

Since $\mathcal{F}$ needs to make two hash queries to get different values of $h_{i}$ and $h_{i}^{'}$, the advantage is $ (\frac{1}{q}\cdot\frac{1}{q-1})$. Furthermore, $\mathcal{F}$ queries the single signature of station $S_{j}$ on the pseudonym $(C_{1}^{*},C_{2}^{*})$ and the message $m^{*}$ with the advantage $\frac{1}{\varrho}$. Therefore, the advantage with which $\mathcal{B}$ can break the CDH assumption is
\begin{equation}
	Adv_{\mathcal{B}}^{CDH}\geq \frac{1}{q}\cdot\frac{1}{q-1}\cdot \frac{1}{\varrho}\cdot \epsilon(l) 
\end{equation}

\begin{theorem}\label{theorem2}
	Our privacy-preserving logistics information systems with traceability (PPLIST) is $(\varrho,\epsilon(l))-$ traceable if the computational Diffie-Hellman (CDH\\) assumption holds on the bilinear group $(e,q,G_{1},G_{2},G_{\iota})$ with the advantage at most $\epsilon_{1}(l)$, the discrete logarithm (DL) assumption holds on the group $G_{2}$ with the advantage at most $\epsilon_{2}(l)$, and $H_{1},H_{2},H_{3}$ are random oracles, where $\varrho$ is the number of signature queries made by the forger $\mathcal{F}$, and $\epsilon(l)=max\big\{( \frac{1}{2}\cdot\frac{1}{q}\cdot \frac{1}{\varrho}\cdot\epsilon_{1}(l)),(\frac{1}{2}\cdot\epsilon_{2}(l))\big\}$.
\end{theorem}
\begin{proof}
	Suppose that there exists an adversary $\mathcal{A}$ that can break the traceability of our scheme, we can construct an algorithm $\mathcal{B}$ which can use $\mathcal{A}$ to break the CDH assumption or DL assumption.
\end{proof}
\paragraph{Setup.} The public parameters are $PUB=(e,q,G_{1},G_{2},G_{\iota},g_{1},g_{2})$.
\paragraph{$\bullet$} $\mathcal{B}$ responds to the queries of $\mathcal{A}$ about the random oracle $H_{1}$.
\begin{itemize}
	\item[1)] $\mathcal{A}$ queries $H_{1}$ on a pseudonym $(C_{1}^{k},C_{2}^{k})$ and a message $m_{k}$. $\mathcal{B}$ selects $t_{k}\stackrel{R}{\leftarrow}Z_{q}$, and sets $g_{1}^{t_{k}}=H_{1}(C_{1}^{k}\parallel C_{2}^{k}\parallel m_{k})$, where $i=1,2,\cdots n$. $\mathcal{B}$ sends $g_{1}^{t_{k}}$ to $\mathcal{A}$, and adds $(C_{1}^{k},C_{2}^{k},m_{k},g_{1}^{t_{k}})$ to the table $Table_{1}$.
	\item [2)]  $\mathcal{A}$ queries $H_{1}$ of a pseudonym $(C_{1}^{*},C_{2}^{*})$ and a message $m^{*}$. $\mathcal{B}$ sets $g_{1}^{\alpha}=H_{1}(C_{1}^{*}\parallel C_{2}^{*}\parallel m_{*})$. $\mathcal{B}$ sends $g_{1}^{\alpha}$ to $\mathcal{A}$, and adds $(C_{1}^{*},C_{2}^{*},m^{*},g_{1}^{\alpha})$ to the table $Table_{1}$.
\end{itemize}
\paragraph{$\bullet$} $\mathcal{B}$ responds to the queries of $\mathcal{A}$ about the random oracle $H_{2}$. Let  $j$ be the index of $Y_{s_{j}}\in AgY$.
\begin{itemize}
	\item [1)] when $i\not=j$, $\mathcal{B}$ selects $h_{i}\stackrel{R}{\leftarrow}\mathbb{Z}_{q}$ and sets $h_{i}=H_{2}(Y_{S_{i}}\parallel AgY)$. $\mathcal{B}$ returns $h_{i}$
	to $\mathcal{A}$, and records it into the table $Table_{2}$. 
	\item [2)] when $i=j$, $\mathcal{B}$ selects $h_{j}\stackrel{R}{\leftarrow}Z_{q}$ and sets $h_{j}=H_{2}(Y_{S_{j}}\parallel AgY)$. $\mathcal{B}$ sends $h_{j}$ to $\mathcal{A}$ and adds $(Y_{S_{j}}, AgY, h_{j})$ into the table $Table_{2}$. 
\end{itemize}
\paragraph{$\bullet$} $\mathcal{B}$ responds to the queries of $\mathcal{A}$ about the random oracle $H_{3}$. When receiving a query $(x_{u},m)$ on $H_{3}$, $\mathcal{A}$ selects $k\stackrel{R}{\leftarrow}\mathbb{Z}_{p}$ and returns it to $\mathcal{A}$. $\mathcal{B}$ adds $(x_{u},m,k)$ into the table $Table_{3}$.

\paragraph{Key-Generation\ Query.}
\begin{itemize}
	\item [1)] $Station-Key-Generation\ Query.$
	$\mathcal{B}$ picks a logistics station $S_{j}$ from $S_{1},S_{2},\cdots,S_{n}$. For the $i$-th logistics station key generation query, $\mathcal{B}$ selects $x_{s_{i}}$, and computes $Y_{S_{i}}=g_{2}^{x_{s_{i}}}$ where $i\not=j$. $\mathcal{B}$ returns $Y_{S_{i}}$ to $\mathcal{A}$. For the $j$-th logistics station key generation query, $\mathcal{B}$ returns $B_{2}$ to $\mathcal{A}$.
	\item [2)] $User-Key-Generation\ Query.$
	$\mathcal{B}$ selects $x_{u}\stackrel{R}{\leftarrow}Z_{q}$, and compute $Y_{u}=g_{2}^{x_{u}}$. $\mathcal{B}$ retains the private key $x_{u}$, sends the public key $Y_{u}$ to $\mathcal{A}$, and sets $GPU$ to store the public key $Y_{u}$ of all users. $Y_{u^{*}}$ is the public key of user $U^{*}$.
	\item [3)] $Trace-Key-Generation\ Query.$
	$\mathcal{B}$ selects $x_{t}\stackrel{R}{\leftarrow}Z_{q}$, and compute $Y_{t}=g_{2}^{x_{t}}$. $\mathcal{B}$ sends the secret-public key pair $(x_{t},Y_{t})$ to $\mathcal{A}$.
\end{itemize}
\paragraph{User-Pseudonym\ Query.}
$\mathcal{A}$ submits a user's secret key $x_{u}$ and a product information $m$. $\mathcal{B}$ first searches in the table $Table_{3}$ and gets $k=H_{3}(x_{u}\parallel m)$. Then $\mathcal{B}$ computes $C_{1}=g_{2}^{k}$, $C_{2}=Y_{t}^{k}\cdot g_{2}^{x_{u}}$. $\mathcal{C}$ sends $(C_{1},C_{2})$ to $\mathcal{A}$.
\paragraph{Public-Key-Aggreation\ Query.}
$\mathcal{A}$ submits a group of logistics stations' public-keys and sets $AgY=\big\{ Y_{s_{1}}, Y_{s_{2}}, \cdots, Y_{s_{d}} \big\}$ . $\mathcal{B}$ searches on the table $Table_{2}$, and gets $h_{i}=H_{2}(Y_{s_{i}}\parallel AgY)$ where $i=1,2,3,\cdots, d$. $\mathcal{B}$ computes $YA=\prod_{i=1}^{d}Y_{s_{i}}^{h_{i}}$. $\mathcal{B}$ sends $YA$ to $\mathcal{A}$. 
\paragraph{Sign\ Query.} $\mathcal{B}$ responds to the queries of $\mathcal{A}$ about the single signature $Sig_{i}$:
\begin{itemize}
	\item [1)] when $i\not=j$, $\mathcal{A}$ asks about the single signature of station $S_{i}$ on a pseudonym $(C_{1},C_{2})\not=(C_{1}^{*},C_{2}^{*})$ and the message $m^{*}$. $\mathcal{B}$ computes $Sig_{i}=g_{1}^{\alpha\cdot h_{i}\cdot x_{s_{i}}}$. $\mathcal{B}$ sends $Sig_{i}$ to $\mathcal{A}$.
	\item [2)] when $i\not=j$, $\mathcal{A}$ asks about the single signature of station $S_{j}$ on a pseudonym $(C_{1},C_{2})$ and a message $m\not=m^{*}$. $\mathcal{B}$ computes $Sig_{i}=g_{1}^{t_{{i}}\cdot h_{i}\cdot \beta}$. $\mathcal{B}$ sends $Sig_{i}$ to $\mathcal{A}$. 
	\item [3)] $\mathcal{A}$ asks about the single signature of station $S_{j}$ on a pseudonym $(C_{1}^{*},C_{2}^{*})$ and the message $m^{*}$. $\mathcal{B}$ aborts.
\end{itemize}
\paragraph{Output.}
$\mathcal{A}$ outputs a final signature $\sigma^{'}$. We consider the following two types of attackers. Suppose that the secret-public key pair of $U^{'}$ is $(x_{u}^{'},Y_{u}^{'})$, and $\mathcal{A}$ only knows $Y_{u}^{'}$. In Type-1 case,  $\mathcal{A}$ outputs a final signature $\sigma^{'}$ containing a new pseudonym $(C^{'}_{1},C^{'}_{2})$. In Type-2 case, $\mathcal{A}$ also outputs a final signature $\sigma^{'}$ which is a signature on a used pseudonym. 

$\emph{Type-1:}$ If there is a new pseudonym $(C^{'}_{1},C^{'}_{2})\notin UPQ$. $\mathcal{C}$ uses $(C^{'}_{1},C^{'}_{2})$ to compute $Y^{'}_{u}=C^{'}_{2}/(C^{'}_{1})^{x_{t}}$, and gets $Y^{'}_{u}$,  $Y^{'}_{u}\notin GPU$. $\mathcal{A}$ forged a single signature $Sig_{i}^{'}$ and a final signature $\sigma^{'}$, where $Sig_{i}^{'}=H_{1}(C^{'}_{1}\parallel C^{'}_{2}\parallel m)^{h_{i}\cdot x_{s_{i}}}$, and $\sigma^{'} = \prod_{i=1}^{d}Sig_{i}^{'}$. Hence, $\mathcal{B}$ can use $\mathcal{A}$ to break CDH assumption, The detailed proof is shown in \emph{Theorem 1}.

$\emph{Type-2:}$ If there is a pseudonym $(C^{'}_{1},C^{'}_{2})$, namely $(C^{'}_{1},C^{'}_{2})\in UPQ$. $\mathcal{B}$ uses $(C^{'}_{1},C^{'}_{2})$ to compute $Y^{'}_{u}=C^{'}_{2}/(C^{'}_{1})^{x_{t}}$, and gets $Y^{'}_{u}$ with $ Y^{'}_{u}\not=Y_{u^{*}}\in GPU$. If $\mathcal{A}$ can prove $C'_{1}=g_{2}^{k'}, C_{2}=Y_{t}^{k'}\cdot g_{2}^{x'_{u}}$. By using the rewinding technique, $\mathcal{B}$ can  extract the knowledge of $(x^{'}_{u},k')$ from $\mathcal{A}$. Therefore,  given $(Y^{'}_{u},g_{2})$, $\mathcal{B}$ can output a $x_{u}^{'}$ such that $Y^{'}_{u}=g_{2}^{x_{u}^{'}}$. Hence, $\mathcal{B}$ can use $\mathcal{A}$ to break the DL assumption.

By the proof of unforgeability, the advantage with which $\mathcal{B}$ can break the CDH assumption is $(\frac{1}{q}\cdot\frac{1}{q-1}\cdot \frac{1}{\varrho}\cdot \epsilon_{1}(l))$. In the situation of \emph{Type-1}, $\mathcal{B}$ can break the CDH assumption with the advantage $ (\frac{1}{2}\cdot\frac{1}{q}\cdot\frac{1}{q-1}\cdot \frac{1}{\varrho}\cdot\epsilon_{1}(l))$. In the situation of \emph{Type-2}, $\mathcal{B}$ can break the DL assumption with the advantage $\frac{1}{2}\cdot\epsilon_{2}(l)$. Hence, $\epsilon(l)=max\big\{( \frac{1}{2}\cdot\frac{1}{q}\cdot\frac{1}{q-1}\cdot \frac{1}{\varrho}\cdot \epsilon_{1}(l)),(\frac{1}{2}\cdot\epsilon_{2}(l))\big\}$.

\textit{\textit{}}\section{Experiment and Evaluation}
In this section, we introduce the implementation and evaluation of our PPLIST scheme.
\subsection{Runtime Environment}
The performance of our PPLIST scheme is measured on a Lenovo Legion Y7000P 2018 laptop with an Intel Core i7-8750H CPU, 500GB SSD and 8GB RAM. The scheme is implemented in Microsoft Windows 10 System using E-clipse Integrated environment, Java language and JPBC library \cite{5983948}.

In our implementation, we apply the Type F curve. For the hash functions $H_{1}:\big\{0,1\big\}^{*}\rightarrow G_{1}$, $H_{2}:\big\{0,1\big\}^{*}\rightarrow Z_{q}$ and $H_{3}:\big\{0,1\big\}^{*}\rightarrow Z_{q}$ required by our scheme, we used $SHA-256$ and the “newElementfromHash()" method in the JPBC library.

Our scheme is implemented in the following three cases: 1) $n=20, d=10$; 2) $n=100, d=50$; 3) $n=200, d=100$. The experimental results are shown in Table~\ref{tab3}.
\begin{table}\centering
	\caption{Times(ms)}\label{tab3}
	\begin{tabular}{|c|c|c|c|}
		\hline
		Phase & n=20,d=10 & n=100,d=50 & n=200,d=100\\
		\hline
		Setup & 522 & 519 & 506 \\
		Station-Key-Generation & 137 & 578 & 1031 \\
		User-Key-Generation & 7 & 5 & 4 \\
		Trace-Key-Generation & 8 & 5 & 3 \\
		User-Pseudonym & 27 & 12 & 12 \\
		Public-Key-Aggregation & 217 & 745 & 1400 \\
		Sign & 122 & 486 & 953 \\
		User-Ownership-Verify & 112 & 106 & 97 \\
		Verify & 176 & 144 & 141 \\
		Trace & 186 & 144 & 141 \\
		\hline
	\end{tabular}
\end{table}
\subsection{Timing}
The setup phase is a process run by the system. It takes 522ms, 519ms and 506ms to setup the system in case 1, case 2 and case 3, respectively. According to the data, it can be observed that the running time of the three cases is roughly the same in the setup phase.

The key pair generation phase is run by logistics stations, user and trace party. It takes 137ms, 578ms and 1031ms to generate the key pair of logistics stations in case 1, case 2 and case 3, respectively. In the user key pair generation phase, it takes 7ms, 5ms and 4ms in case 1, case 2 and case 3, respectively. For trace party to generate key pair, it takes 8ms, 5ms and 3ms in case 1, case 2 and case 3, respectively. 

The pseudonym generation phase is run by the user. It takes 27 ms, 12ms and 12 ms in case 1, case 2 and case 3, respectively. Observing the experimental data, it is not difficult to find that the running time of the public key aggregation phase is proportional in the number of logistics stations. It takes 259ms, 745ms and 1400ms to aggregate public keys in the three cases, respectively. The signature phase is run by logistics stations. The times to generate a multi-signature in case 1, case 2 and case 3 are 122ms, 486ms and 953ms, respectively.  

In the user ownership verification phase, a user proves the ownership by interacting with the last logistics station. It takes 112ms, 106ms and 97ms in case 1, case 2 and case 3, respectively. In signature validation phase, it takes 186ms, 144ms and 141ms to verify a multi-signature in case 1, case 2 and case 3, respectively. To trace a user, it takes 176ms, 144ms and 141ms in case 1, case 2 and case 3, respectively. We implement our scheme in three different cases. The experiment results show the efficiency of our scheme. 
\section{Conclusions}
In this paper, to protect users' privacy in LIS, a PPLIST was proposed. In our scheme, users anonymously use logistics services. Furthermore, a trace party can de-anonymized users to prevent illegal logistics. Additionally, the whole logistics process can be recorded and is unforgeable. We formalize the definition and security model of our scheme, and present a formal construction. We formally proved the security of our scheme and implemented it.

In our scheme, a buyer can prove the ownership of a product by proving the knowledge included in the pseudonyms. Our future work is to improve the flexibility of this work to enable an owner to designate a proxy to prove the ownership of products on behalf of him.

\section*{Acknowledgment}
This work was partially supported by the National Natural Science Foundation of China (Grant No. 61972190, 62072104, 61972095) and the National key research and development program of China (Grant No. 2020YFE0200600). This work was also partially supported by the Postgraduate Research \& Practice Innovation Program of Jiangsu Province (Grand No. KYCX20\_1322) and the Natural Science Foundation of the Fujian Province, China (Grant No. 2020J01159).
\bibliographystyle{splncs04}
\bibliography{a}
\end{document}